\documentclass{article}

\usepackage{fullpage}
\usepackage{amsmath}
\usepackage{amssymb}
\usepackage{amsthm}
\usepackage[all]{xy}
\usepackage{array}
\usepackage{graphicx,color}
\usepackage{comment}
\usepackage{framed}
\usepackage{etoolbox}
\usepackage{pstricks}
\usepackage{epic,eepic}
\usepackage{wrapfig}
\usepackage{float}
\usepackage{subfig}
\usepackage[noend]{algpseudocode}
\usepackage{algorithm}

\newtheorem{thm}{Theorem}[section]
\newtheorem{lemma}[thm]{Lemma}
\newtheorem{prop}[thm]{Proposition}
\newtheorem{cor}[thm]{Corollary}
\newtheorem{definition}[thm]{Definition}

\begin{document}

\title{Persistent Homology Over Directed Acyclic Graphs}

\author{Erin Wolf Chambers \footnote{
    Department of Computer Science,
    Saint Louis University,
    \{echambe5,letscher\}@slu.edu.
    Research supported in part by the National Science Foundation under Grant No. CCF 1054779 and IIS-1319573.
  }
  \and David Letscher \footnotemark[1]
}

\date{}

\maketitle

\begin{abstract}
We define persistent homology groups over any set of spaces which have inclusions 
defined so that the corresponding directed graph between the spaces is acyclic, as well as along any
subgraph of this directed graph. 
This method simultaneously generalizes standard persistent homology, zigzag persistence and multidimensional persistence to arbitrary directed acyclic graphs, and it also allows the study of more general families of topological spaces or point-cloud data. 
We give an algorithm to compute the persistent homology groups simultaneously for all subgraphs which contain a single source and a single sink in $O(n^4)$ arithmetic operations, where $n$ is the number of vertices in the graph. 
We then  demonstrate as an application of these tools a method to overlay two distinct filtrations of the same underlying space, which allows us to detect the most significant barcodes using considerably fewer points than standard persistence.
\end{abstract}

\section{Introduction}
\label{sec:intro}

Since its introduction~\cite{persistence}, the concept of topological persistence  has found numerous 
applications in diverse areas such as surface reconstruction, sensor networks, bioinformatics, 
and cosmology.  
For completeness, we briefly survey some results from persistent homology with an 
emphasis on tools and techniques used in this paper, although a full coverage is beyond 
the scope of this paper.  See any of the recent books or surveys on topological 
persistence for full coverage of this broad topic and its applications~\cite{eh-phs-08,eh-cti-10,ghristoverview,z-tc-05,oudot2015persistence}.

At a high level, the various models of persistence all have a collection of spaces and inclusions of one space into another.
The aim is to find topological features common to subsets of these spaces.  In standard persistent homology~\cite{persistence},
the spaces are linearly ordered with one space included in the next.  In zigzag persistence\cite{zigzag}, the spaces are 
linearly ordered but the inclusions can occur in either direction.  
Multi-dimensional persistence~\cite{multid-persistence} works in multiple dimensions on a grid with inclusion maps 
parallel to the coordinate axes.  In this paper, we present a natural extension of each of these to a directed acyclic graph (or DAG) structure on the spaces and maps, which we call DAG persistence.
See Figure~\ref{fig:models} for an example of each of thesestructure.

Standard persistence consider spaces with maps of the form ${X}_1 \rightarrow {X}_2 \rightarrow \ldots \rightarrow {X}_n$, where each $X_i$ is a topological space, often represented as a simplicial complex.
These maps between the spaces induce maps between chain complexes which pass to homology as 
homomorphisms $H({X}_1) \rightarrow H({X}_2) \rightarrow \ldots \rightarrow H({X}_n)$.
Persistent homology identifies homology classes that are ``born'' at a certain location in the filtration
and ``die'' at a later point.  These identified cycles encompass all of the homological information in the
filtration and have a module structure~\cite{module}.  This persistence module has a unique decomposition into a sum of ``elementary'' modules,
which are intervals that share a common homological feature.  The endpoints of these intervals are the birth and death times.
The set of intervals gives a barcode representation of the persistence module.
Persistent homology algorithms have been implemented very efficiently, 
since the homology groups with coefficient in a finite field form vector spaces, and the inclusion maps induce linear 
maps between the spaces.

\begin{figure}\centering
\begin{tabular}{lclc}
    (a) & \small
     $\xymatrix@R=1.4pc @C=1.4pc{
      X_0 \ar[r] & X_1 \ar[r] & X_2 \ar[r] & X_3 \ar[r] & X_4 \\
    }$ &
    (b) & \small
    $\xymatrix@R=1.4pc @C=1.4pc{
      X_0 \ar[r] & X_1 & X_2 \ar[l] & X_3 \ar[l]\ar[r] & X_4 \\
    }$ \\[.25in]
 
 	(c) & \small
    $\xymatrix@R=1.4pc @C=1.4pc{
      X_{02} \ar[r] &  X_{12} \ar[r] & X_{22} \ar[r]] & X_{32} \\
      X_{01} \ar[r]\ar[u] &  X_{11} \ar[r]\ar[u] & X_{21} \ar[r]\ar[u] & X_{31} \ar[u] \\
      X_{00} \ar[r]\ar[u] &  X_{10} \ar[r]\ar[u] & X_{20} \ar[r]\ar[u] & X_{30} \ar[u] \\
    }$ &
    (d) & \small
    $\xymatrix@R=1.4pc @C=1.4pc{
      & X_1 \ar[rr] \ar[dr] & & X_5 \\
      & & X_3 \ar[r] \ar[dr] & X_6 \\
      & X_2 \ar[ur] \ar[r] & X_4 \ar[ur] \ar[r] & X_7  \\
    }$
\end{tabular}
\caption{\label{fig:modesl} The underlying graph structure for (a) Standard persistence, (b) zigzag persistence, 
(c) Multi-dimensional persistence and (d) DAG persistence.}
\label{fig:models}
\end{figure}

Zigzag persistence considers spaces with maps of the form ${X}_1 \leftrightarrow {X}_2 
\leftrightarrow \ldots \leftrightarrow {X}_n$, where the maps can go in either direction.  
These maps between the spaces induce maps between chain complexes which pass to homology as 
homomorphisms $H({X}_1) \leftrightarrow H({X}_2) \leftrightarrow \ldots \leftrightarrow H({X}_n)$; 
this is known as a zigzag module.  The zigzag module has the same structure as the persistence module:
it has a unique decomposition into a sum of ``elementary'' modules, which are intervals that share a common homological feature.
Like standard persistence, these intervals give a barcode representation of the zigzag module.
Recent work in this setting includes an algorithm which examines the 
order of the necessary matrix multiplications quite carefully and is able to get a running time for 
a sequence of $n$ simplex deletions or additions which is dominated by the time to multiply two 
$n \times n$ matrices~\cite{zigzag-mm}.

In multi-dimensional persistence, the spaces $\{X_u\}_{u \in \mathbb{Z}^d}$ lie on an integer lattice with inclusion
maps for $X_u \to X_v$ where $u, v \in \mathbb{Z}^d$ differ in a single coordinate.
Multi-dimensional persistence modules have a more complicated structure than standard or zigzag, so its interpretation is far more difficult.
In particular, no barcode representation exists for multi-dimensional persistence.
One of the primary tools is the rank invariant, $\rho_{X,k}(u,v)$, which measures the rank of homology groups in common among
all $X_w$ with $u_i \le w_i \le v_i$.

\paragraph{Our contribution}
In this paper, we give a generalization of persistence to spaces where the underlying inclusions 
form a directed acyclic graph (DAG).  This simultaneously generalizes both zigzag and multidimensional 
persistence, which can be viewed as special cases of these underlying 
graphs on the maps between the spaces.  

To define persistence modules over DAGs we expand upon the use of quiver modules, first applied to topological persistence in the original zigzag persistence paper~\cite{zigzag}, to incorporate commutativity conditions that arise from following different paths in a DAG between the same pair of vertices.  
Note that the authors in~\cite{Escolar2016} examined other ways of incorporating commutativity conditions in specific families of graphs.
Given a persistence module, we utilize techniques from category theory, namely limits and co-limits, to define persistence groups for any subgraph of a DAG.

We then give algorithms to compute the persistent homology for DAGs in various settings.  
In Section~\ref{sec:single}, for graphs with at most $n$ vertices, 
we give an $O(n^4)$ algorithm for  computing the persistent homology groups 
for all single-source single sink subgraphs of a DAG simultaneously.  In Section~\ref{sec:general}, we describe theoretical algorithms for computing persistence homology for general subgraphs.

Potential applications of this are extensive, including any spaces where inclusions are more 
general than previous settings. We present two such  applications in 
Section~\ref{sec:applications}.
The first uses multiple samples of the same space to accurately find
significant topological features with far fewer sample points than other
methods require.  The second application uses DAG persistence to measure the similarity
between two spaces.

\section{Definition}
\label{sec:def}

We recall some relevant definitions and background before presenting our
definition of persistent homology over directed acyclic graphs.  For a full presentation 
of homology groups see any introductory text in algebraic topology, e.g.~\cite{hatcher,munkres}.  

For a simplicial complex $X$ and an Abelian group $A$, a $k$-chain is a formal linear
combination of the $k$-simplices of $X$ with coefficients from $A$; these 
$k$-chains form a group which we denote $C_k(X,A)$ (where we will 
generally omit the $A$ if the group is clear from context).  The map $\partial_k: C_k(X) \to C_{k-1}(X)$ is a
linear map that calculates the boundary of a chain.  The cycle
group is defined as $Z_k(X) = \{ c \in C_k(X) \ | \ \partial_k(c) = 0 \} = ker(\partial_k)$ 
and the
boundary group is the group $B_k(X) = \{ c \in C_k(X) | \ \exists d \in C_{k+1}(X)$ with
$\partial_{k+1} (d) = c \} = im (\partial_{k+1})$.  The homology group
is defined as $H_k(X) = Z_k(X)/B_k(X)$.  Note that if $A$ is a field then
$C_k(X), Z_k(X), B_k(X)$ and $H_k(X)$ are all vector spaces.

Given a filtration $X_0 \subset X_1 \subset \cdots \subset X_n$, the persistent
homology group $H_k^p(X_j)$ can be defined in multiple ways.  Traditionally,
it is defined as $Z_k(X_j) / B_k(X_j) \cap Z_k(X_{j+p})$ and can be viewed
as a quotient group of $H_k(X_{j+p})$.  An equivalent definition is
$H_k^p(X_j) = im(i_*)$, where $i_*: H_k(X_j) \to H_k(X_{j+p})$ is the map
induced by the inclusion $i:X_j \to X_{j+p}$.  We note then that
$H_k^p(X_j)$ can also be thought of as a subgroup of $H_k(X_{j+p})$.

\subsection{Graph Filtrations}

In a sense, the persistence group for some interval of a filtration can be thought as the subgroups
which are common to all of the homology groups.  This motivates the definition of 
persistence groups over more general sets of inclusion maps; indeed, zigzag persistence and multi-dimensional persistence are examples of this type of generalization.  
In this paper, we generalize to consider inclusions over a set of spaces that form a  directed  graph.  
The main restriction we will place on this graph is that it must 
be acyclic and not contain repeated edges, which is a natural constraint for any graph which represents a set of inclusions.  
Note that we could have equivalently defined these filtrations using a poset, but prefer to use the graph theory terminology.
More formally:
\begin{definition}
For a simple directed acyclic graph $G = (V,E)$, a \emph{graph filtration} $\mathcal{X}_G$ of a topological space
$X$ is a pair $( \{ X_v \}_{v \in V}, \{f_e\}_{e \in E} )$ such that
\begin{enumerate}
  \item $X_v \subset X$ for all $v \in V$
  \item If $e = (v_1, v_2) \in E$ then $f_e: X_{v_1} \to X_{v_2}$ is a continuous embedding (or inclusion)
  of $X_{v_1}$ into $X_{v_2}$.
  \item The diagram commutes: in other words, suppose we have a path $\gamma = e_1, \ldots, e_l$, we can naturally extend this to a function on the topological spaces $f_\gamma = f_{e_k} \circ \cdots \circ f_{e_1}$.  Then given $\gamma$ and $\gamma'$ which are two different directed paths connecting vertices $u$ and $w$, commuting means that $f_\gamma = f_{\gamma'}$. 
\end{enumerate}
\end{definition}

\begin{figure}\centering\small
$\xymatrix@R=1.4pc @C=1.4pc{
  & X_1 \ar[rr] \ar[dr] & & X_5 \\
  & & X_3 \ar[r] \ar[dr] & X_6 \\
  & X_2 \ar[ur] \ar[r] & X_4 \ar[ur] \ar[r] & X_7  \\
}$ 
$\xymatrix@R=1.4pc @C=1.4pc{
  & H_k(X_1) \ar[rr] \ar[dr] & & H_k(X_5) \\
  & & H_k(X_3) \ar[r] \ar[dr] & H_k(X_6) \\
  & H_k(X_2) \ar[ur] \ar[r] & H_k(X_4) \ar[ur] \ar[r] & H_k(X_7)  \\
}$
\caption{A graph filtration over a graph $G$ and the corresponding commutative $G$-module.}
\label{fig:filtration}
\end{figure}

\subsection{Persistence Module}

Before defining the persistent homology groups for a directed acyclic graph, we will first generalize the persistence module~\cite{module} and zigzag perisistence module~\cite{zigzag}.  
We will use an approach similar to that used in the definition of the zigzag persistence module that utilizes quivers and their representations~\cite{barot2015representations,miyachi2000representations}.
See~\cite{oudot2015persistence} for a thorough discussion of the connections between quivers, their representations and persistent homology.

A \emph{quiver} is a directed graph where loops and multiple edges between the same vertices are allowed.
So every DAG is also a quiver.
Given a quiver $G$, a \emph{representation} of $G$ imparts every vertex of $G$ with a vector space and every edge a linear map between the vector spaces at its endpoints.  
A representation of a quiver $G$ is also referred to as a $G$-module.

In order to include the commutativity conditions that are part of a graph filtration, we must use a \emph{quiver with relations}.  
Formally, a quiver with relations is defined in terms of an ideal of an associative algebra associated to the quiver that is called the path algebra, see \cite{barot2015representations}.
However, a representation of a quiver with relations can be defined without the use of a path algebra.
We will focus only on commutativity relations, but more general relations can also be dealt with in a similar manner.
This motivates the definition of a commutative $G$-module, which is a representation of a quiver with (commutativity) relations; see Figure~\ref{fig:filtration} for an example of such a space.

\begin{definition}
For a directed acyclic graph $G = (V,E)$, a \emph{commutative $G$-module} is the pair $( \{ W_v \}_{v \in V}, \{ f_e \}_{e \in E}$)
where for each vertex $v$, $W_v$ is a vector space and for any edge $e = (v,w)$, $f_e: W_v \to W_w$ is a linear map with the condition
that the resulting diagram is commutative.  
\end{definition}

This definition provides the framework for discussing the persistence module for a graph that extends the
definition for the zigzag persistence module, which will be shown in Section~\ref{subsec:relations}.

\begin{definition}
For a directed acyclic graph $G = (V,E)$ and $k$-dimensional persistence module for a graph filtration $\mathcal{X}_G$, $\mathcal{PH}_k(\mathcal{X}_G)$, is the commutative $G$-module $( \{ W_v \}_{v \in V}, \{ f_e \}_{e \in E}$ where
	\begin{itemize}
    	\item $W_v = H_k(X_v)$ for all $v \in V$
        \item For every edge $(u,v) \in E$, $f_e: H_k(X_u) \to H_k(X_v)$ is the map induced by the inclusion $X_u \to X_v$.
	\end{itemize}
\end{definition}

The theory for commutative $G$-modules is very similar to that of zigzag persistence modules.  
A commutative $G$-module $\mathcal{V} = ( \{ V_v \}, \{ g_e \} )$ is a \emph{submodule} of $\mathcal{W} = (\{ W_v \}, \{ f_e \})$ if $V_v \subset W_v$ for all $v$ and $f_e|_{V_v} = g_e$.
Similarly, given two commutative $G$-modules $\mathcal{V} = ( \{V_v\}, \{f_e\})$ and $\mathcal{W} = ( \{W_v\}, \{g_e\})$,
we can define their connected sum $\mathcal{V} \oplus \mathcal{W}$ as $( \{V_v \oplus W_v\}, \{f_e \oplus g_e\})$.
A commutative $G$-module is said to be \emph{indecomposable} if it cannot be written as a non-trivial connected
sum.  
The commutative $G$-modules that we are considering are finite, so  $\mathcal{V}$, can be \emph{decomposed} as $\mathcal{V} = \mathcal{V}_1 \oplus \cdots
\oplus \mathcal{V}_n$, where each $\mathcal{V}_i$ is indecomposable.  
For a connected subgraph $G'$ of $G$,
we will define the commutative $G$-module $\mathbb{F}_{G'}$ as the module with a copy of $\mathbb{F}$ at each
vertex of $G'$ and zero elsewhere; we will put the identity map on each edge of $G'$ and make every other map 
trivial. We will call this module \emph{elementary}.

\begin{thm}[Krull-Remak-Schmidt Theorem for Finite-Dimensional Algebras \cite{lang}]
The decomposition of a commutative $G$-module is unique up to isomorphism and permutation of the summands.
\end{thm}

In the case of standard and zigzag persistence, the relevant indecomposable modules are always elementary.  This is
implied by Gabriel's theorem~\cite{quiver} which provides an enumeration of indecomposable modules
for particular graph types.  Figure \ref{fig:indecompossible} gives an example of an indecomposable
module that is not elementary.  The example consists of a sphere with four punctures, and the inclusion
of each of the four boundary components.  Unfortunately, the existence of such examples tells us that
there is no simple ``barcode'' representation for DAG persistence.  In Section~\ref{subsec:barcodes},
we will generalize barcodes for our context.

\begin{figure}
\centering
\begin{tabular}{lclclc}
(a) & \subfloat{\small$\vcenter{\hbox{\subfloat{\small\xymatrix@R=1pc @C=1pc{
A \ar[rd] & & B \ar[ld] \\
& S & \\
C \ar[ur] & & D \ar[ul]
}}}}$} &
(b) &
\subfloat{$\vcenter{\hbox{\includegraphics[width=.75in]{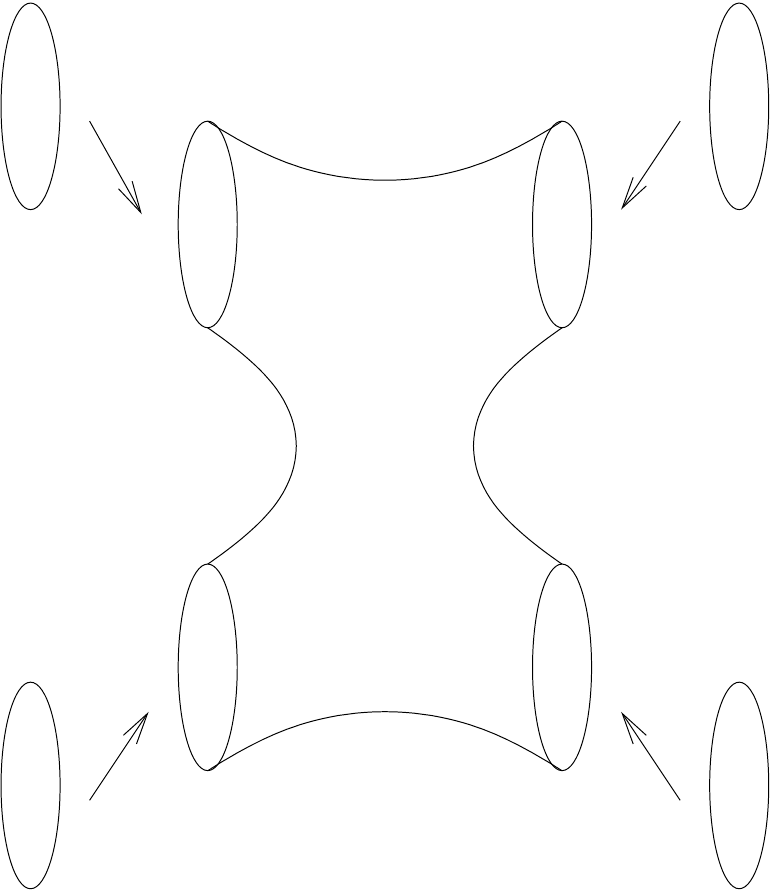}}}$}
& (c) &
\subfloat{\small$\vcenter{\hbox{\subfloat{\small\xymatrix@R=1pc @C=1pc{
\mathbb{F} \ar[rd] & & \mathbb{F} \ar[ld] \\
& \mathbb{F}^3 & \\
\mathbb{F} \ar[ur] & & \mathbb{F} \ar[ul]
}}}}$}
\end{tabular}
\caption{(a) The persistence module for the graph $G$ is not finite-type; there are infinitely many non-isomorphic irreducible $G$-modules.  
(b) A graph filtration $\mathcal{X}_G$
(c) The persistence module for $\mathcal{X}_G$ which is 
irreducible and non-elementary.
}
\label{fig:indecompossible}
\end{figure}

\subsection{Persistent Homology}

We will define persistent homology groups for any connected subgraph of $G$ to included homological features that are present in the homology groups at every vertex of the subgraph.  
To make this precise, we will define the persistence groups in terms of the limit and co-limit of a commutative $G$-module; see any text on category theory for a complete discussion of limit and co-limits, for example~\cite{mac2013categories}.

In category theory, limits and co-limits are defined for diagrams, a very general concept.
In particular, commutative diagrams with every vertex having an object and every edge a morphism are diagrams in the sense of category theory~\cite{mac2013categories}.  In particular, commutative $G$-modules are diagrams.  We will give the definitions of limits and co-limits and related structures in terms of commutative $G$-modules; however, these definitions are identical to those for arbitrary diagrams in \cite{mac2013categories}. 

The \emph{cone} of the commutative $G$-module is a pair, $(L,\phi)$, with a vector space $L$ and homomorphisms $\phi_v: L \to H_k(X_v)$, such that for any edge $(u,v)$, $f_{uv} \circ \phi_u = \phi_v$.  
The \emph{limit} of a commutative $G$-module is a cone $(L,\phi)$ such that for any other cone $(L',\phi')$, there is a unique homomorphism $u: L' \to L$ such that $\phi_e \circ n = \phi'_e$ for every edge $e$ of $G$.  
Similarly, a \emph{co-cone} of a commutative $G$-module is a pair $(C,\psi)$, with a vector space $C$ and homomorphism $\psi_v: H_k(X_v) \to C$ such that for any edge $(u,v)$, $\psi_v \circ f_{uv} = \psi u$; the \emph{co-limit} is a co-cone $(C,\psi)$ of the commutative $G$-module such that for any other co-cone $(C',\psi')$ there is a unique homomorphism $u: C \to C'$ such that $u \circ \psi_e = \psi'_e$ for every edge $e$ of $G$.  See Figure~\ref{fig:definition} for an illustration of the limit and co-limit.
When they exist, limits and co-limits are unique up to isomorphism~\cite{mac2013categories}.  
The existence of the limits and co-limits relies on a few results from category theory, all can be found in~\cite{mac2013categories}.
First, $G$-modules are in the category of vector spaces which are known to be bi-complete, which means limits and co-limits exist for diagrams from what are know as small categories.  Furthermore, commutative diagrams are small categories.  Together these imply the existence of limits and co-limits of commutative $G$-modules.
We will denote the limit and co-limit of $M$, by $lim(M)$ and $colim(M)$, and let $\mu_M: lim(M) \to colim(M)$ be the induced map between them.

\begin{definition}
\label{def:module_persistence}
If $M$ is a commutative $G$-module then the \emph{persitence} of $M$, $\mathcal{P}(M)$, is image $\mu_M(lim(M))$.
\end{definition}

The persistence of a module represents all features, or subspaces, common to every vector space $M_v$ in the commutative $G$-modules.   
This will be made precise in Lemma~\ref{lem:ssss}.  

\begin{definition}
\label{def:persistence}
Given a graph filtration $\mathcal{X}_G$ and a connected subgraph $G' \subset G$, the $G'$-\emph{persistent homology group}, $H_k^{G'}(\mathcal{X}_G)$, is $\mathcal{P}\left( \mathcal{PH}_k(\mathcal{X}_G) | _{G'} \right)$, where $\mathcal{PH}_k(\mathcal{X}_G) | _{G'}$ is the module $\mathcal{PH}_k(\mathcal{X}_G)$ restricted to the subgraph $G'$.
\end{definition}

In Section~\ref{subsec:relations}, we connect this definition of persistence to known models and prove that it simultaneously generalizes standard persistence,
zigzag persistence and multidimensional persistence. 

\begin{figure}\centering\small
\begin{displaymath}\xymatrix@R=1.4pc @C=1.4pc{
  & & H_k(X_1) \ar[rr] \ar[dr] & & H_k(X_5) \ar@{.>}[drr]\\
  lim(\mathcal{PH}_k(\mathcal{X}_G))  \ar@{.>}[urr] \ar@{.>}[rrr] \ar@{.>}[drr]
  & & & H_k(X_3) \ar[r] \ar[dr] & H_k(X_6) \ar@{.>}[rr] & & colim(\mathcal{PH}_k(\mathcal{X}_G)) \\
  & & H_k(X_2) \ar[ur] \ar[r] & H_k(X_4) \ar[ur] \ar[r] & H_k(X_7) \ar@{.>}[urr] \\
}\end{displaymath}
\caption{The limit, $L_{\mathcal{X}_G} $, and co-limit, $C_{\mathcal{X}_G} $, of a diagram of homology groups.
(For clarity, only maps from the limit to source vertices and maps from sink vertices to the co-limit are shown.)}
\label{fig:definition}
\end{figure}

\subsection{Barcodes and Carrier Subgraphs}
\label{subsec:barcodes}

We wish to generalize the notion of a persistence barcode to the DAG setting.  
Every submodule of $\mathcal{PH}(\mathcal{X}_G)$ has a subgraph $G'$ where the module is non-trivial.
We will call this the \emph{carrier subgraph} of that module and say that the module is
\emph{carried} by that subgraph.
In standard and zigzag persistence all irreducible
modules are elementary, and a barcode is precisely a carrier subgraph of this module.
In our more general setting, we cannot assume that irreducible modules are elementary.
In fact, these irreducible submodules can be very complicated.
For example, in Figure~\ref{fig:indecompossible} the carrier subgraph would be the entire graph.
Even though carrier subgraphs for general DAGs are not a complete invariant, that is not all information is encoded in these representations, they can still provide insight into the structure of the persistence module.

\subsection{An Example}
\label{subsec:ex}

In Figure~\ref{fig:example}, we see an example of a set of spaces with inclusions that form a 
directed acyclic graph.   At the top level is a genus two surface; the directed arrows indicate the 
inclusion maps in our directed acyclic graph, down to our two source vertices in the graph which 
include one space with two disjoint annuli and one space that is a disk with three additional boundaries.
The graph forms a poset that demonstrates the non-trivial intersections and unions of the three 
surfaces.  In Figure~\ref{fig:example}b, we see the persistence module for the entire 
space, and in Figures~\ref{fig:example}c and d we see the indecomposable submodules and their carrier subgraphs.
From these indecomposables it is possible to read off the persistence for any subgraph $G'$
by counting how many of the elementary modules have $G'$ as a subgraph.

\begin{figure}
\centering
\begin{tabular}{lclc}
(a) & $\vcenter{\hbox{\includegraphics[width=2in]{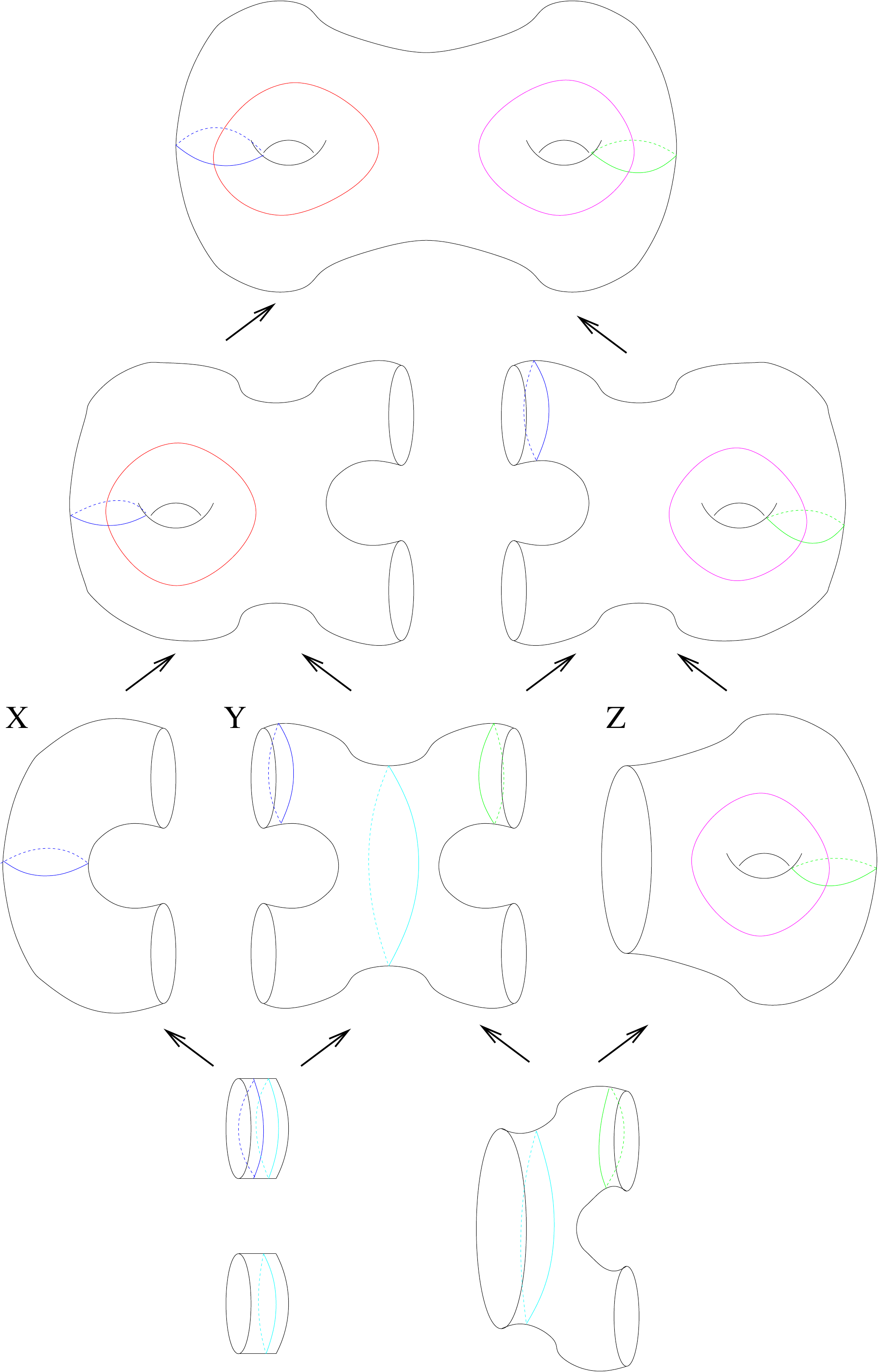}}}$ &
(b) & $\vcenter{\hbox{
{\xymatrix@R=2pc @C=2pc{
 & \mathbb{F}^4 & \\
  \mathbb{F}^3 \ar[ur] & & \mathbb{F}^3 \ar[ul] \\
  \mathbb{F} \ar[u] & \mathbb{F}^3 \ar[ul]\ar[ur] & \mathbb{F}^2 \ar[u] \\
  \mathbb{F}^2 \ar[u]\ar[ur] & & \mathbb{F}^2 \ar[ul]\ar[u]
}}}}$\\
\vspace*{.25in} & \\
(c) & $\vcenter{\hbox{\begin{tabular}{cc}
{\tiny\xymatrix@R=.75pc @C=.75pc{
  & \mathbb{F} & \\
  \mathbb{F} \ar[ur] & & \mathbb{F} \ar[ul] \\
  \mathbb{F} \ar[u] & \mathbb{F} \ar[ul]\ar[ur] & {\color{gray} 0} \ar@{.}\ar@{.}[u] \\
  \mathbb{F} \ar[u]\ar[ur] & & {\color{gray} 0} \ar@{.}[ul]\ar@{.}[u]
}} & 
 {\tiny\xymatrix@R=.75pc @C=.75pc{
  & \mathbb{F} & \\
  \mathbb{F} \ar[ur] & & \mathbb{F} \ar[ul] \\
  {\color{gray} 0} \ar@{.}[u] & \mathbb{F} \ar[ul]\ar[ur] & \mathbb{F} \ar[u] \\
  {\color{gray} 0} \ar@{.}[u]\ar@{.}[ur] & & \mathbb{F} \ar[ul]\ar[u]
}} \\
{\tiny\xymatrix@R=.75pc @C=.75pc{
  & \mathbb{F} & \\
  \mathbb{F} \ar[ur] & & {\color{gray} 0} \ar@{.}[ul] \\
  {\color{gray} 0} \ar@{.}[u] & {\color{gray} 0} \ar@{.}[ul]\ar@{.}[ur] & {\color{gray} 0} \ar@{.}[u] \\
  {\color{gray} 0} \ar@{.}[u]\ar@{.}[ur] & & {\color{gray} 0} \ar@{.}[ul]\ar@{.}[u]
}} &
 {\tiny\xymatrix@R=.75pc @C=.75pc{
  & \mathbb{F} & \\
  {\color{gray} 0} \ar@{.}[ur] & & \mathbb{F} \ar[ul] \\
  {\color{gray} 0} \ar@{.}[u] & {\color{gray} 0} \ar@{.}[ul]\ar@{.}[ur] & \mathbb{F} \ar[u] \\
  {\color{gray} 0} \ar@{.}[u]\ar@{.}[ur] & & {\color{gray} 0} \ar@{.}[ul]\ar@{.}[u]
}} \\
{\tiny\xymatrix@R=.75pc @C=.75pc{
  & {\color{gray} 0} & \\
  {\color{gray} 0} \ar@{.}[ur] & & {\color{gray} 0} \ar@{.}[ul] \\
  {\color{gray} 0} \ar@{.}[u] & \mathbb{F} \ar@{.}[ul]\ar@{.}[ur] & {\color{gray} 0} \ar@{.}[u] \\
  \mathbb{F} \ar@{.}[u]\ar[ur] & & \mathbb{F} \ar[ul]\ar@{.}[u]
}} 
\end{tabular}}}$ &
(d) & $\vcenter{\hbox{\includegraphics[width=1.6in]{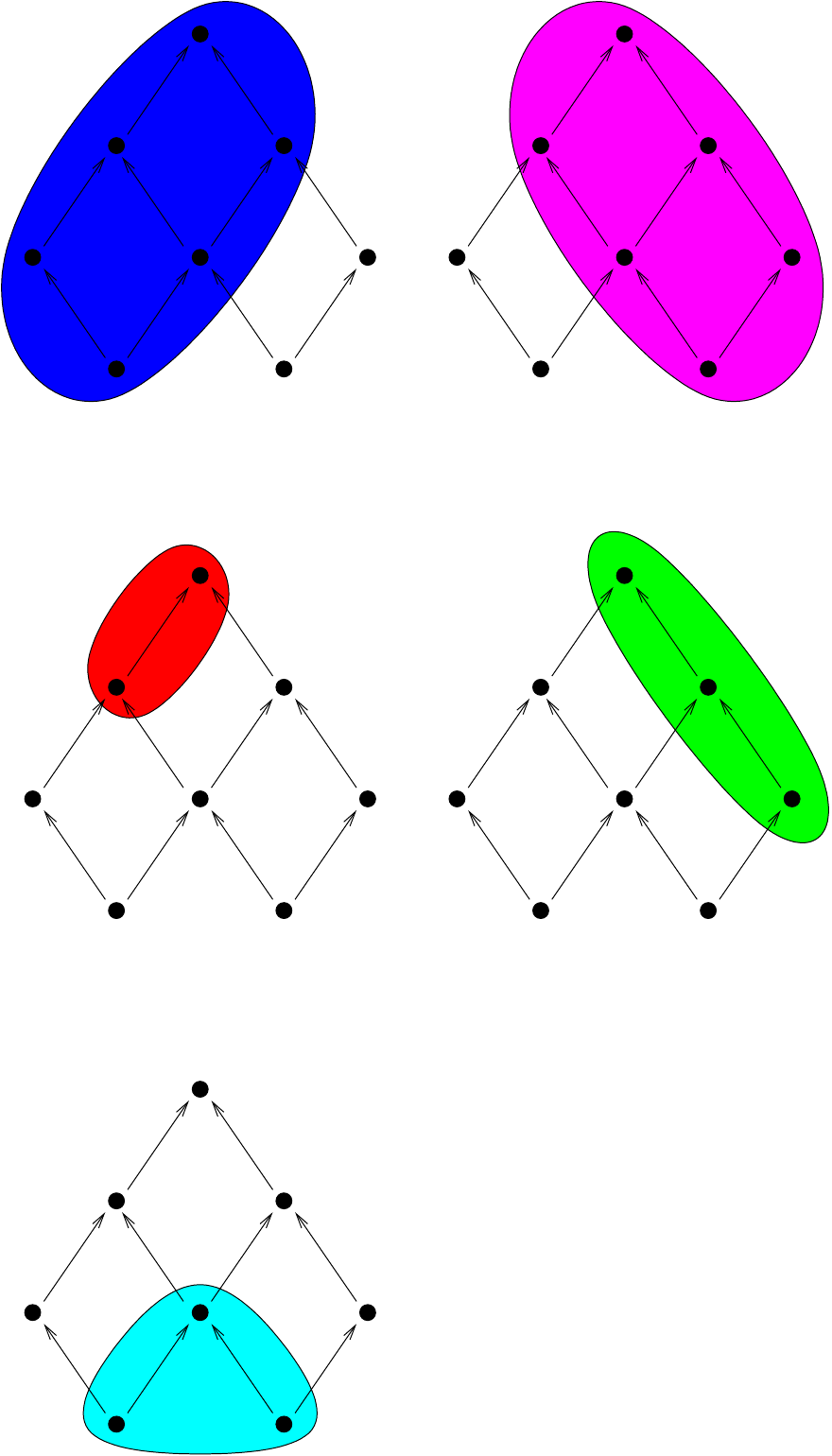}}}$ \\
\end{tabular}
\caption{(a) A genus two surface that is the union of three subsurfaces ($X$, $Y$ and $Z$), with maps between subspaces forming a DAG, (b) its persistence
module $\mathcal{PH}_1(\mathcal{X}_G)$, (c) the five indecomposable summands of the persistence module (indicated in bold)
and (d) the five carrier subgraphs (marked by ovals and shaded different colors) for this example.}
\label{fig:example}
\end{figure}

\section{Properties of Persistence Module}

Before examining algorithms to calculate persistence and persistence modules over DAGs, we will examine various properties of commutative $G$-modules and their relationship with the persistence module and zigzag persistence module.

\subsection{Fundamental Properties}

There are several properties of commutative $G$-modules and their persistence that are useful in seeing their relationships with existing models of persistence and how they can be computed.  The first shows how the persistence of a commutative $G$-module, $M$, carries information common to every vertex group $M_v$.

\begin{lemma}
If $M$ is a commutative $G$-module then $$M \cong \mathcal{P}(M)_G \oplus N$$
where $\mathcal{P}(M)_G$ is the module with a copy of $\mathcal{P}(M)$ at each vertex of $G$ and $\mathcal{P}(N) = 0$.
\end{lemma}
\begin{proof}
Consider the linear map $\mu_m: lim(M) \to colim(M)$.
This gives two decompositions
\begin{eqnarray*}
 lim(M) &\cong& im(\mu_M) \oplus ker(\mu_m) \ = \ \mathcal{P}(M) \oplus ker(\mu_m) \\
 colim(M) &\cong& im(\mu_M) \oplus cok(\mu_m) \ = \ \mathcal{P}(M) \oplus cok(\mu_m)
\end{eqnarray*}
where $cok(\mu_M) = N/im(\mu_m)$.
So there is a copy of $\mathcal{P}(M) \subset lim(M)$  that maps isomorphically to $\mathcal{P}(M) \subset colim(M)$.  
Consider the limit of $M$, $\mathcal{P}(M)$, and the maps $\phi_v: \mathcal{P}(M) \to M_v$.
This decomposes $M_v \cong P \oplus cok(\phi_v)$.
Notice that the maps $M_u \to M_v$ preserve this decomposition.
This established that $M \cong \mathcal{P}(M)_G \oplus N$, where
$N_v \cong cok(\phi_v)$.

To show that $\mathcal{P}(N) = 0$, note that $\mathcal{P}(M_1 \oplus M_2) \cong \mathcal{P}(M_1) \oplus \mathcal{P}(M_2)$, since limits and co-limits preserve direction sums~\cite{mac2013categories}.
If  $\mathcal{P}(N) \neq 0$ then 
$$\mathcal{P}(M) \cong \mathcal{P}(\mathcal{P}(M)|_G) \oplus \mathcal{P}(N)
\cong \mathcal{P}(M) \oplus \mathcal{P}(N) \not\cong \mathcal{P}(M)$$
This contradiction shows that the persistence of $N$ is trivial.
\end{proof}

Note that the previous implies that for every vertex $v$ of $G$, there is an injection of $\mathcal{P}(M)$ into $M_v$.
This lemma yields a characterization of when indecomposable submodules are elementary; precisely when their persistence is non-trivial.

\begin{cor}
\label{lem:modulesinglesource}
Assume $M$ is an indecomposable commutative $G$-module then
$\mathcal{P}(M) \neq 0$ if and only if $M \cong \mathbb{F}_G$.
\end{cor}
\begin{proof}
Using the decomposition $M \cong \mathcal{P}(M)_G \oplus N$ in the previous lemma, we observe that since $M$ is irreducible one of the terms must be trivial.  In one case $M = \mathcal{P}(M)_G = \mathbb{F}_G$ and in the other $\mathcal{P}(M) = 0$.
\end{proof}

The next lemma provides the basis of both computing DAG persistence on a large class of subgraphs and relating this model of persistence to existing theories.

\begin{lemma}
\label{lem:ssss}
If $G$ is a single-source single-sink graph with source and sink vertices $s$ and $t$, respectively, and $M = (\{M_V\},\{f_e\})$ is a commutative $G$-module then $\mathcal{P}(M) \cong im\left(M_s \to M_t\right)$
\end{lemma}
\begin{proof}
First consider the limit of $M$.  
Notice that for any cone $(C,\{g_v\})$ of $M$, all of the maps $g_v: C \to M_v$ are completely determined by $g_s$; specifically $g_v = f_{\gamma_{sv}} \circ g_s$ where $\gamma_{sv}$ is any path in $G$ from $s$ to $v$.
We claim that $(M_s, \{f_{\gamma_{sv}}\})$ is the limit of $M$.
$$\small\xymatrix@R=1.4pc @C=2.4pc{
M_s \ar[dr]^{1} \ar@/^1pc/[drrr]_{f_{\gamma_{sv}}} \ar@/^2pc/[drrrrr]_{f_{\gamma_{st}}} \\
& M_s \ar[r] & \ldots \ar[r] & M_v \ar[r] & \ldots \ar[r] & M_t \\
C \ar[ur]^{g_s} \ar@/_1pc/[urrr]^{f_{\gamma_{sv}}\circ g_s} \ar@/_2pc/[urrrrr]^{f_{\gamma_{st}} \circ g_s} \ar@{-->}[uu]^{g_s}
}$$
In the above diagram it is clear that the cone from $C$ factors through the cone $M_s$.

A similar argument shows that the co-limit of $M$ is $(M_t,\{f_{\gamma_{vt}}\})$.
So, $\mathcal{P}(M)$, the image of the limit in the co-limit is identical to the image
$M_s \to M_t$.
\end{proof}

%

\subsection{Relationship With Other Models of Persistence}
\label{subsec:relations}

We saw in Section~\ref{sec:def} the definition for standard persistence.  Multidimensional
persistence can also be defined in terms of the image of a map.  Consider a multifiltration,
or $d$-dimensional grid of spaces equipped with a partial ordering on vertices, where
$u = (u_1,\ldots, u_d) \le v = (v_1, \ldots, v_d)$ if and only if $u_i \le v_i$ for
all $i$.  %

$$\xymatrix@R=1pc @C=1pc{
X_{x} \ar[r] & \cdots\ar[r] & X_{v} \\
\vdots \ar[u] & & \vdots \ar[u] \\
X_{u} \ar[r]\ar[u] & \cdots\ar[r] & X_{y} \ar[u] \\
}$$

The \emph{rank invariant}, $\rho_{X,k}(u,v)$, is defined as the dimension of the
image of the induced map $H_k(X_u) \to H_k(X_v)$~\cite{zigzag}.  Since the diagram commutes,
this map can be found following any path from $u$ to $v$ in the graph.

Zigzag persistence is defined in terms of the zigzag module.  The definition of
a commutative $G$-module and a $\tau$-module for zigzag persistence are identical
for a zigzag graph.  The following proposition specifies how DAG persistence
generalizes these three notions of persistence.

\begin{prop}
\label{prop:comparison}
Suppose $\mathcal{X}_G$ is a graph filtration of $X$.  Then:

\begin{enumerate} 
  \item (Standard persistence) If $G$ is the graph corresponding to the filtration
    $X_0 \to X_1 \to \cdots \to X_n$ and $I_{i,p}$ is the subgraph consisting of vertices
    $\{X_i, \ldots, X_{i+p}\}$ then $H_k^{I_{i,p}}(\mathcal{X}_G) \cong H_k^p(X_i)$.
    Furthermore, $\mathcal{PH}_k(\mathcal{X}_G)$ coincides with the persistence module. 
    
  \item (Zigzag persistence) If $G$ is the graph for zigzag persistence
    $\mathbb{X} = X_0 \leftrightarrow X_1 \leftrightarrow \cdots \leftrightarrow X_n$ where 
    each arrow could go in either direction,
    then $\mathcal{PH}_k(\mathcal{X}_G)$ coincides with the zigzag persistence module.
    
  \item (Multidimensional persistence) Let $\mathcal{X} = \{ X_v \}_{v \in \{0,\ldots,m\}^d}$ be a multifiltration with
    underlying graph $G$.  If $G_{u,v}$ is the subgraph with vertices $\{ w \in G \ | \ u \le w \le v \}$ then
    the rank invariant $\rho_{X,k}(u,v) = \dim H_k^{G_{u,v}}(\mathcal{X}_G)$.
\end{enumerate}
\end{prop}

\begin{proof}
In the case of zigzag persistence, the definition of the zigzag persistence module is identical to the definition of the DAG persistence module; in particular for the case of a zigzag graph $G$ the definition of a commutative $G$-module coincides exactly with the definition of a $\tau$-module~\cite{zigzag}.  
Note that the standard persistence module is a special case of the zigzag persistence module so this also shows that the persistent module coincides with the DAG persistence module when the graph is a single path.

In the case of standard persistence, Lemma~\ref{lem:ssss}
implies that $H_k^{I_i,p}(\mathcal{X}_G) = im(H_k(X_i) \to H_k(X_{i+p})$, which is one of the definitions of the (standard) persistent homology group.  So the two ways of calculating persistent homology groups coincide for graphs that are paths.

In the case of multidimensional persistence, $\rho_{X,k}(u,v)$ is defined as the rank of the image $H_k(X_u) \to H_k(X_v)$.  Lemma~\ref{lem:ssss} can also be applied in this situation to show that $H_k^{G_{u,v}} = im(H_k(X_u) \to H_k(X_v))$.  So, the rank invariant is the same as the rank of the (DAG) persistent homology group.
\end{proof}

\section{Single-Source Single-Sink Subgraphs}
\label{sec:single}

In this section, we consider $\mathcal{X}_G$ where $G$ is a directed 
acyclic graph with a single source vertex $s$.  Also, each vertex of the graph will represent a subset of a particular fixed cell complex, although this final complex may or may not actually appear as a complex attached to a vertex in $G$. 
To limit the number of cells which $f_e$ can introduce (where $e = (u,v) \in G$), 
we assume that each inclusion $f_e$ adds a single cell to the underlying space.  
We will also assume that at each source vertex $s$ that $X_s = \emptyset$.
These assumptions are standard in most persistence algorithms~\cite{zigzag-mm,module} 
and quite natural given that we can decompose any inclusion map into a series
of inclusions of one simplex at a time.

If we consider a subgraph $G'\subset G$ with source $s$ and sink $t$, Lemma~\ref{lem:ssss} implies that we only need to find the image of
$H_k(X_s) \to H_k(X_t)$.  This image can be found by following any path
from $s$ to $t$ and applying the standard persistence algorithm.
This can be repeated for every pair of vertices of $G$.
Therefore, a straightforward application of the cubic time
standard persistence algorithm would yield an $O(n^2 l^3)$ algorithm, where $n$ is the number of vertices
of the graph and $l$ is the length of the longest directed path from source to sink.  Using tools from 
recent work to compute zigzag persistence in matrix multiply time~\cite{zigzag-mm} 
would give a running time of $O(n^2 (M(l))+l^2 \log^2 l))$, where $l$ is the length of the longest path 
between any source and sink and $M(l)$ is the time to multiply two $l \times l$ matrices.

Here, we will adapt the standard persistence algorithm~\cite{persistence}, shown in Algorithm~\ref{alg:standard}, to calculate the 
persistent homology for all single-source single-sink subgraphs.  First, we recall
 the standard persistence algorithm, which starts with the $(k+1)$-dimensional boundary map stored in a matrix.  
 This means that there is a column for each $(k+1)$-cell and a row for each $k$-cell.
Each column stores the boundary of its $(k+1)$-cell, with a $1$ or $-1$ in each entry, where the sign
depends on the positive or negative orientation of the $k$-cell as a face in the larger $(k+1)$-cell.
There is also an additional
array that stores the index of the first non-zero entry in each column.  (In the original
paper, it was assumed that the field was $\mathbb{Z}_2$, but the same algorithm works
for any finite field.)  For completeness, below is pseudocode for the standard algorithm
applied to a boundary matrix $B$.  Note that the entries for the matrix $\{b_{cr}\}$ are stored in column major order and that the $i$-th column of $B$ will be denoted by $B_i$.

\begin{algorithm}
\caption{Standard Persistence}
\begin{algorithmic}
\Procedure{StandardPersistence}{$B$}
\For{$c=1..\#$ columns of $B$}
	\For{$c'=1..c-1$}
    	\State Let $r$ be the first row where $b_{c'r}$ is non-zero
    	\If{$b_{cr} \neq 0$}
    		\State Replace $B_c$ with $B_c - b_{cr}(b_{c'r})^{-1} B_{c'}$
        \EndIf
    \EndFor
\EndFor
\State \textbf{return} $B$
\EndProcedure
\end{algorithmic}
\label{alg:standard}
\end{algorithm}
At its completion, the algorithm terminates with a matrix where the index each column is the death time of a cycle that has birth time at the index of its first non-zero row.  And a basis for $H_k^p(X_i)$ can be extracted from $B$ extracting the columns $1, \ldots, i+p$ whose first non-zero entry has index at least $i$~\cite{persistence}.

Our algorithm to calculate persistence on every single-source single-sink subgraph of $G$ relies on two observations.  The first is that the standard persistence algorithm can be modified to run on a tree without increasing its asymptotic runtime.  The second, is that every DAG can be covered in a linear number of trees so that for any vertices $u$ and $v$ that have a directed path between them in $G$ also have a path between them in at least one of the trees.  So we can call the tree based persistence algorithm $n$ times and recover all of the persistent homology groups for single-source single-sink subgraphs of $G$.

To calculate persistence on a tree, $T$, we make the same assumptions as before, at the root $r$, $X_r = \emptyset$ and the transition from each edge add a single cell.  Let $O$ be an ordered list of the vertices of $T$ in any topological ordering and let $<_T$ be the partial order for the elements of the tree.  The boundary matrix $B$ will be an $|O|\times|O|$ matrix with both rows and columns indexed by $O$.  The entries of the column $B_v$ will be the boundary of the cell $\sigma$ that is added in the inclusion $X_u \to X_v$ where $(u,v)$ is an edge of $T$.  Notice that if $T$ was just a path, the matrix $B$ would be the same as the boundary matrix used in the standard persistence algorithm.  Below is the algorithm.

\begin{algorithm}
\caption{Persistence Calculation on a Tree}
\label{alg:tree}
\begin{algorithmic}
\Procedure{TreePersistence}{$B$,$O$,$<_T$}
\For{$c \in O$}
	\For{$c' \in O$}
    	\If {$c <_T c'$}
        	\State Let $r \in O$ be the first row with $b_{c'r}$ non-zero
          	\If{$b_{cr} \neq 0$}
              	\State Replace $B_c$ with $B_c - b_{cr}(b_{c'r})^{-1} B_{c'}$
          	\EndIf
      	\EndIf
    \EndFor
\EndFor
\State \textbf{return} $B$, $R$
\EndProcedure
\end{algorithmic}
\end{algorithm}

If $\gamma \subset T$ is a path from $u$ to $v$ then a basis for $H_k^\gamma(\mathcal{X}_G)$ can be found by extracting the column from $B$ with column index at most $v$ in the partial ordering $<_T$ and that have their first non-zero entry at least $u$ in the partial ordering.

\begin{lemma}
If $\{\gamma_v\}$ are paths the from the root $r$ to the vertex $v$ in $T$ then
Algorithm~\ref{alg:tree} correctly finds a basis $H_k^{\gamma_v}(\mathcal{X}_G)$ for all $k$ and all $v \in T$ in $O(n)$ time, where $n$ is the number of vertices of $G$.
\end{lemma}
\begin{proof}
Notice that entries in column $v$ and row $u$ of $B$ are always $0$ if $u$ and $v$ are not connected by a directed path in $T$.  
In the calculation of column $v$ of the matrix, only columns $u$ with $u \le_T v$ can affect the values in the column.  
If the calculations involving the other columns are ignored, then the calculation is identical to the one done in the standard persistence algorithm (with extraneous zeros) and returns the same result.

The algorithm clearly runs in $O(n^3)$ time since each loop proceeds through at most $n$ elements and the body of the inner loop takes linear time to scan for the first non-zero elements and update the row.
\end{proof}

To calculate persistence for all single-source single-sink subgraphs, we construct $n$ trees, one for each vertex $v$ in the graph.  
The tree starts with a path from a source vertex to $v$ combined with a tree that spans all vertices $u$ where there is a path $u$ to $v$.  
Since the persistent homology for these graphs does not depend on the path taken, running the tree algorithm on this tree will always calculate persistence for all single-sink single-source subgraphs with source $v$.

\begin{algorithm}
\caption{Persistent Homology for All Single-Source Single-Sink Subgraphs}
\label{alg:ssss}
\begin{algorithmic}
\Procedure{AllSingleSournceSingleSink}{$\mathcal{X}_G, V, E, <_G$}
	\For {$v \in V$}
    	\State Let $\gamma$ be a path from any source to $v$
    	\State $U = \{ u \in V  \ | \ v \le_G u \} \cup \gamma$
    	\State Let $T \subset G $ be a tree rooted at $v$ with vertex set $U$
        \State Let $O$ be a topological ordering of the vertices of $T$
        \State Construct a square matrix $B$ of all zeros indexed by the elements of $O$
        \For {$v \in O$}
        	\State Let $(u,v)$ be an edge of $T$
        	\State Let $\sigma$ be the simplex of $X_v - X_u$
            \State Set the column indexed by $v$ to store the entries of $\partial\sigma$
        \EndFor
        \State $B_v =$ \Call{TreePersistence}{$B, O, <_T$}
    \EndFor
    \State \textbf{return} $\{B_v\}$
\EndProcedure
\State\textbf{Note:}
The tree $T$ is a spanning tree for the subgraph of $G$ with vertex set $U$, the partial ordering and partial order can be constructed in $O(n \log n)$ time at the same time a topological ordering can be constructed using standard graph algorithms, see~\cite{cormen2009introduction} for details.
\end{algorithmic}\end{algorithm}

\begin{thm}
\label{thm:SSSS}
If $G$ is a directed acyclic graph that has $n$ vertices then Algorithm~\ref{alg:ssss} calculates all of the persistent homology groups with coefficients from a finite field for all
single source-single sink subgraphs of $G$ can be calculated in 
$O(n^4)$ time and $O(n^3)$ space.
\end{thm}
\begin{proof}
Consider a single-source single-sink subgraph $G'$ with source $u$ and sink $v$.  Lemma~\ref{lem:ssss} implies that $H_k^{G'}(\mathcal{X}_G) \cong im(H_k(X_u) \to H_k(X_v))$, which we know that Algorithm~\ref{alg:tree} correctly calculates and stores in $B_u$.

Algorithm~\ref{alg:ssss} loops through each vertex once.  
The body of the loop takes $O(n \log n)$ time to construct the tree and associate structures and $O(n^3)$ time for the call to Algorithm~\ref{alg:tree}.  This gives a total runtime of $O(n^4)$.  
Notice that each call to tree persistence uses $O(n^2)$ space so the totral space usage is $O(n^3)$.
\end{proof}

We note that if we remove the assumption that coefficients are from a finite field, our theorem
still accurately counts the number of arithmetic operations, although running times for each 
operation might take longer.  In 3 dimensions, however, no information is 
lost when persistent homology is taken with respect to a finite field~\cite{persistenthomotopy},
so this is not an overly restrictive assumption.
We also note that this algorithm can be adapted to 
multidimensional persistence to give an improvement over the known polynomial time 
algorithm~\cite{computingmdp}:

\begin{prop}
\label{prop:lattice}
All of the rank invariants for a $d$-dimension lattice with $n$ nodes can be calculated
in $O(n^{4-1/d})$ time.
\end{prop}

\begin{proof}
The improvement in the run-time for lattices is that we do not need to run the tree based algorithm for every vertex of the graph.  
Instead we can choose trees for each vertex of the form $v = (0,v_2,\ldots,v_d)$.  
These trees will follow a path from the origin to $v$ and then go through every vertex $u = (u_1,u_2,\ldots,u_d)$ with $v_i < u_i$ for $i\ge 2$.  
In addition to the edges in the path to $v$, there will be an edge from $u' \to u$ whenever $u'$ is the largest vertex in the lexiographic order with $v \le_G u' \le_G u$.

{\hfill
\centering\xymatrix@R=.75pc @C=.75pc{
	\bullet & \bullet & \bullet & \bullet & \bullet \\
	\bullet \ar[u] & \bullet \ar[u] & \bullet \ar[u] & \bullet \ar[u] & \bullet \ar[u] \\
	\bullet \ar[u] & \bullet \ar[u] & \bullet \ar[u] & \bullet \ar[u] & \bullet \ar[u] \\
	\bullet \ar[u]\ar[r] & \bullet \ar[u]\ar[r] & \bullet \ar[u]\ar[r] & \bullet \ar[u]\ar[r] & \bullet \ar[u] \\
    \bullet \ar[u] & \circ & \circ & \circ & \circ \\
    \bullet \ar[u] & \circ & \circ & \circ & \circ \\
}\hfill}

\vspace{.1in}
Every pair of vertices that can be connected in the lattice by a directed path can also be connected in one of these trees.  
There are $n^\frac{d-1}{d} = n^{1-1/d}$ such
trees, yielding the improved running time.
\end{proof}

\section{Indecomposible Submodules and Persistence Over General Subgraphs}
\label{sec:general}

For a general subgraph $G' \subset G$, calculating $\mathcal{PH}_k^{G'}(\mathcal{X}_G)$ is much more complicated.  
However, we can utilize algorithms from computational algebra to solve calculate DAG persistence for arbitrary subgraphs.
To do so, we first need to address decomposing the DAG persistence module into indecomposables.

\paragraph{Algorithm to decompose the persistence module}
Viewed in other contexts a DAG persistence module is a representation of a quiver with (commutative) relations or, equivalently, a representation of an associative algebra.
Our representations are finite-dimensional and have coefficients in a finite field.  
Under these conditions there are polynomial time algorithms to decompose the module into indecompossibles~\cite{chistov1997polynomial,brooksbank2008testing}.  These, and similar, algorthms have been implemented in the computer algebra systems GAP~\cite{GAP4} and MAGMA~\cite{MR1484478}.  For examples that are not too large, these algorithms work well in practice.  For larger examples that we attempted, GAP failed to return an answer.

\paragraph{Algorithm to calculate the rank of $\mathcal{H}_k^{G'}(\mathcal{X}_G)$ for general subgraphs}
Given a decomposition into indecomposable submodules, it is relatively straightforward to calculate the persistence groups fro arbitrary subgraphs.
\begin{algorithmic}
\State Let $P$ be the restriction of  $\mathcal{PH}_k(\mathcal{X}_G)$ to the subgraph $G'$.
\State Decompose $P$ into irreducibles $M_1,\ldots, M_n$.
\State Initialize $r=0$
\For{$i=1..n$}
    \If{$M_i$ is elementary} \State Increment $r$ \EndIf
\EndFor
\end{algorithmic}  
Given a representation of a submodule, we can check if it is elementary since at every vertex of the DAG the representation is either rank 0 or 1.  
Note that Lemma~\ref{lem:modulesinglesource} implies that each of the $M_i$ are either elementary and contribute exactly one to the rack of the persistence group or have $H_k^{G'}(M_i) = 0$ and contribute nothing to the persistence group.

\section{Applications}
\label{sec:applications}

We finally will  demonstrate two proof of concept applications of our definitions, both of which use
our single source single sink algorithm from Section~\ref{sec:single}.  Both examples
utilize the same dataset:  a piecewise analytic surface of genus two. Point samples were 
sampled uniformly from the surface.

\subsection{Estimating Persistence Using Multiple Subsamples}

The first application is estimating persistence for a point sample using
a pair of much smaller subsamples.  Let $X_0 \to X_1 \to \cdots \to X_n$
and $Y_0 \to Y_1 \to \cdots \to Y_n$ be filtrations for the union of balls
of various radii for two subsamples of a common point set.  Moreover, we assume that each of the $X_i$
is contained in some $Y_j$ and vice-versa.  This yields a directed
graph of the form shown in Figure~\ref{fig:parallel_graph}.
This is the picture of an $\epsilon$-interleaving~\cite{cohen2008proximity}.
It is know that the persistence diagrams of these two filtrations have distance at most $\epsilon$~\cite{chazal2016structure}.
Using DAG persistence we try to exploit additional information encoded in the graph to distinguish the significant cycles from noise.

If we could calculate the carrier graphs for each irreducible submodule for such a dataset, then we could build a ``persistence diagram''. 
The birth times would be the average index of the first $X_i$ and $Y_i$ that appear in the carrier graph and the death time would be the average of the largest index.  
In theory, we could use the algorithms in Section~\ref{sec:general}, but they were too large for GAP to process.
Instead, we used an alternative heurstic process:

\begin{enumerate}
	\item Fix a small constant $s > 0$.
	\item Build filtrations $\{ X_i \}$ and $\{ Y_i \}$ where 
    	$X_i = \{ x \ | \ d(x,X_0) \le i s \}$ and $Y_i$ is defined similarly.
	\item Estimate if $X_i \subset Y_j$ and $Y_i \subset X_j$.
    	\begin{enumerate}	
        	\item Decompose the region into a rectangular grid with spacing $s$.
        	\item Each grid point is labeled by its distance to $X_0$ and $Y_0$
            	which were rounded to their closest multiple of $s$.
            \item We will assume that $X_{\alpha s} \subset Y_{\beta s}$ if
            	for every grid points with distance at most $\alpha s$ to $X_a$
                has distance to $Y_0$ at most $\beta s$.
        \end{enumerate}
    \item Calculate the rank of the persistent homology using the single-source
    	single-sink algorithm of Section \ref{sec:single} and record their bases.
    \item Consider the set of bases and merge any two graphs that have one of the
    	same bases in common.
\end{enumerate}

This heuristic algorithm attempts to construct the carrier subgraphs of the irreducbile submodules.
This process is not guaranteed to identify the carrier subgraphs of the irreducible submodules.  
However, in our particular example, we were able to validate the results and show that it correctly identified the carrier subgraphs.
We find it interesting that the heuristic was successful in finding the irreducible submodules and wonder how likely this is to occur in practice.

In Figure~\ref{fig:parallel}, we show the result of this process.  
Each subfigure is a persistence diagram, where each cycle is considered as a pair of
birth and death times.  The persistence of each cycle is the difference in these times
which is equal to the distance to the diagonal.
The 
space considered is a genus two surface sampled with 5000 points, at which level the 
persistent features, which are the 4 generators of homology, are clearly seen as significant in Figure~\ref{fig:parallel}(a).
Note that these four points blur together in pairs in the figure and are difficult to distinguish from each other.  In 
the remaining pictures, we calculate persistent homology for a simple directed acyclic graph that 
consists of two different 200 point subsamples (Figures~\ref{fig:parallel}(b) and (c)), and their union into a larger 400 point sample (Figure~\ref{fig:parallel}(d)).  We note 
that individually, each sample's persistent homology is quite noisy and does not distinguish the four 
generators at all from the ``noise''.  However, the persistent homology for the directed acyclic graph (Figure ~\ref{fig:parallel}(e)) clearly separates the 4 main generators from the noise, at a far lower level of sampling 
than is possible with standard persistent homology.  The table (Figure~\ref{fig:parallel}(f)) shows the persistence values
for the four generating cycles for the surface and the next largest cycle.  For simplicity of interpretation, the input data was scaled so that the most significant cycle in the 5,000 point sample has lifespan equal to 100.

\begin{figure}
\centering
\begin{tabular}{c}
{\xymatrix@R=.75pc @C=.75pc{
  X_0 \ar[r]  \ar[ddrrr] & 
  X_1 \ar@{.>}[rr] \ar[ddrrrr] & &
  X_i \ar@{.>}[rr] & &
  X_j \ar@{.>}[rr] & &
  X_n
  \\ \\
  Y_0 \ar[r] \ar[uurrr] & 
  Y_1 \ar@{.>}[rr] \ar[uurrrr] & &
  Y_i \ar@{.>}[rr] & &
  Y_j \ar@{.>}[rr] & &
  Y_n  
}}
\end{tabular}
\caption{\label{fig:parallel_graph} Filtration of union of balls corresponding to each
of the two subsamples where all arrows indicate inclusion maps.}
\end{figure}

Intuitively, short lived cycles in one of the filtrations do not align well with cycles in the other filtration.  
So these cycles are paired with very short cycles in the other filtration.  
The corresponding points on the persistence diagram are moved toward the diagonal.  
This deemphasizes the cycles that are the result of sampling noise.

Additional possible applications of this are numerous. 
 For example, we could use directed acyclic graphs 
of various witness complexes and seek improved results using a DAG over a small set of witness 
complexes; recent work using zigzag persistence seems relevant in this setting~\cite{zigzagapplications}.  
Also, the example in Section~\ref{subsec:ex} demonstrates how homology classes from pieces of a
space can be ``aligned'' similar to the bootstrapping method of~\cite{zigzagapplications}.

It should be noted, that this is not a faster way to calculate the persistence of the 
larger sample.  This experiment was performed to test if the filtrations for two small subsamples
could be combined to more accurately represent the topology of the shape than the union of the
subsamples.  In this case, we were able to answer in the affirmative, but this would not work on all such examples and would be computationally infeasible in practice.

\begin{figure}
\begin{center}
\hspace*{-.4in}
\begin{tabular}{ccc}
(a) Full (5000 points) & (b) First half (200 points) & (c) Second half (200 points) \\
$\vcenter{\hbox{\includegraphics[width=2.2in]{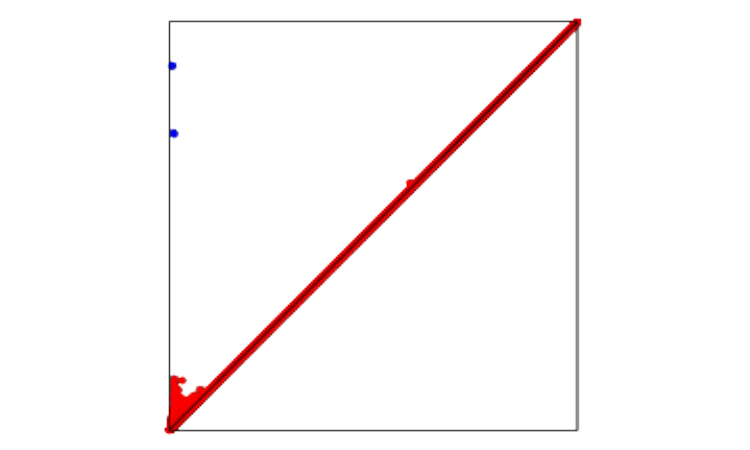}}}$ &
$\vcenter{\hbox{\includegraphics[width=2.2in]{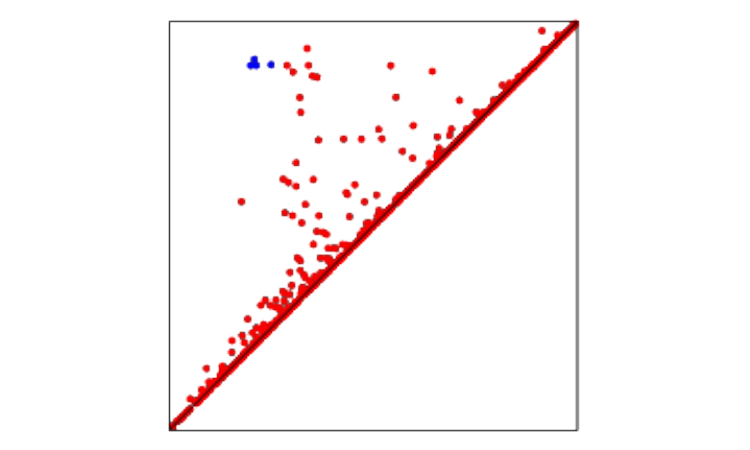}}}$ &
$\vcenter{\hbox{\includegraphics[width=2.2in]{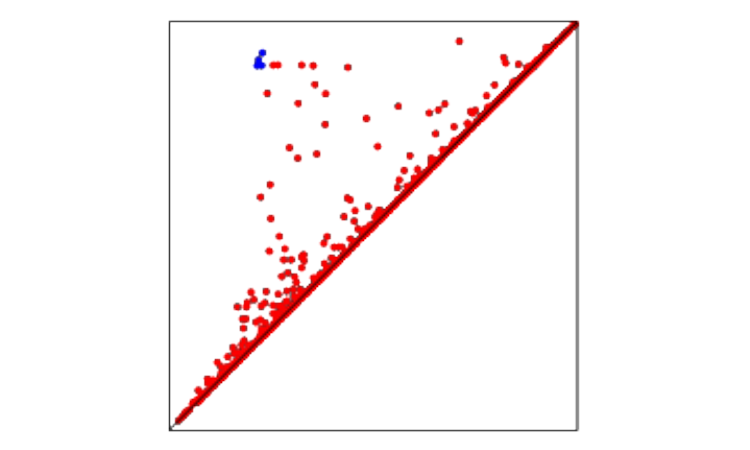}}}$ \\
\vspace*{.05in} & & \\
(d) Merged (400 points) & (e) DAG persistence & (f) Persistence values \\
$\vcenter{\hbox{\includegraphics[width=2.2in]{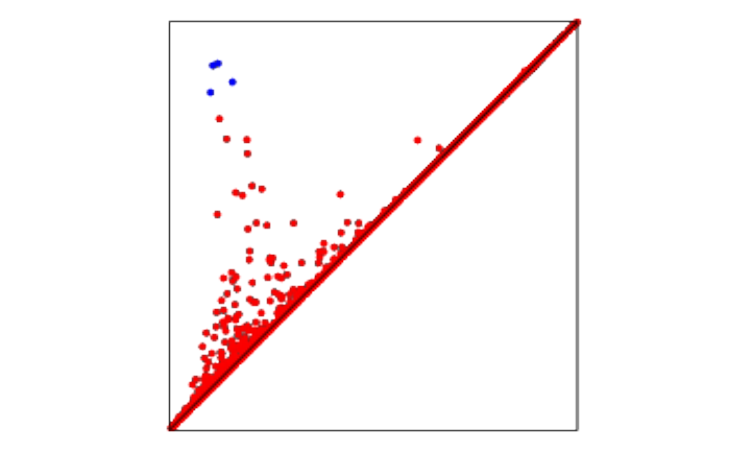}}}$ &
$\vcenter{\hbox{\includegraphics[width=2.2in]{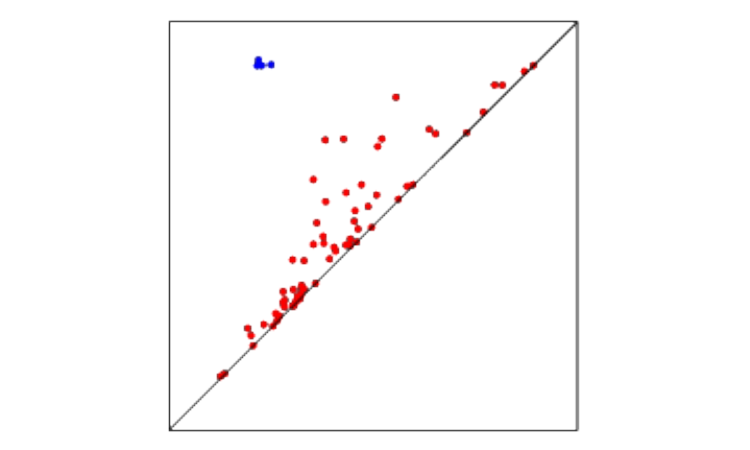}}}$ &
{\small
\begin{tabular}{l|c|c}
& Top 4 & Fifth \\ \hline
Full & 80.4--100.0 & 12.8\\
Subsample & 78.8--88.7 & 72.1 \\
1st half & 72.8--78.9 & 68.2 \\
2nd half & 75.2--78.5 & 72.1 \\
Parallel & 72.8--77.6 & 37.1 \\
\end{tabular}}
\end{tabular}
\end{center}
\caption{Persistent homology for a point sample of a double torus.
(a) Persistence diagram for the original 5,000 point sample,
(b) and (c) Persistence diagrams for two different 200 point subsamples,
(d) Persistence diagram for the union of the two small samples,
(e) DAG persistence diagram
and (f) the distance from the diagonal to the five most significant features.  In all of the persistence diagrams, the four most significant cycles are colored blue to highlight them.}
\label{fig:parallel}
\end{figure}

\begin{figure}
\centering
\begin{tabular}{c}
{\small\xymatrix@R=.75pc @C=.75pc{
& & & & \\
& & & & \\
& & X_2 \cup Y_2 \ar@{.>}[uu] & \\
X_2 \ar[urr]\ar@{.>}[uuu] & & & Y_2 \ar[ul]\ar@{.>}[uuu] \\
& X_2 \cap Y_2 \ar[ul]\ar[urr] \ar@{.>}[uuuu] & & \\
& & & & \\
& & X_1\cup Y_1 \ar[uuuu] & \\
X_1 \ar[urr]\ar[uuuu] & & & Y_1 \ar[ul]\ar[uuuu] \\
& X_1 \cap Y_1 \ar[ul]\ar[urr]\ar[uuuu] & & \\
& & & & \\
& & X_0\cup Y_0 \ar[uuuu] & \\
X_0 \ar[urr]\ar[uuuu] & & & Y_0 \ar[ul]\ar[uuuu] \\
& X_0 \cap Y_0 \ar[ul]\ar[urr]\ar[uuuu] & & \\
}}
\end{tabular}
\caption{\label{fig:comparison}Underlying graph for the comparison of two filtrations.}
\end{figure}

\subsection{Shape Comparison}
\label{subsec:comparison}

Given filtrations of two overlapping shapes, a natural question
is to measure how similar they are.  DAG persistence can provide
a method for comparison.  Consider the graph in Figure~\ref{fig:comparison};
it includes filtrations for two shapes $X$ and $Y$ as well as their
intersection and union.  If $X$ and $Y$ are very similar and well
aligned then this should be detected in the carrier subgraphs.

This comparison is made by building persistence diagrams for
each filtration $\{ X_i \}$ and $\{ Y_i \}$ using standard persistence.
A persistence diagram is built for the module over $G$ by first 
calculating its irreducible modules and their carrier subgraphs.  In this example,
all of the irreducible subgraphs were elementary making further analysis possible.  
Each carrier subgraph is converted to
a point $(i,j)$ in the persistence diagram if $i$ is the smallest
index and $j$ is the largest index of the vertices in the carrier subgraph.

Similar to the previous application, it was not possible to perform this calculation exactly as we
have no efficient algorithm to find the irreducible submodules.
Instead we used a very similar heuristic process to the one in the previous example, except
we had four filtrations to compare instead of two.

This comparison was performed on  two 1000 point subsamples from 
the dataset in the previous example, see Figure~\ref{fig:parallel}.  The bottleneck distance 
between the persistence diagrams of each of the two samples
and the persistence diagram for the comparison graph were both under
5\% of the lifespan of the 4 significant topology features of the shapes.
This demonstrates that the two shapes being compared have nearly
identical topological features except on a very small scale.

An alternate way of comparing the filtration $\{X_i\}$ and $\{Y_i\}$ is to find the interleaving distance between them~\cite{cohen2008proximity}.  This is the minimum $\epsilon$ so that for all $i$, $X_i \subset Y_{i+\epsilon}$ and $Y_i \subset X_{i+\epsilon}$.  In general, calculating interleaving distance is very hard to do.  The DAG persistence for this graph does not approximate interleaving distance, but it has the potential to be used for some of the same purposes.
This is a direction for further study.

\section{Future Work}
\label{sec:future}

Our algorithm in Section \ref{sec:general} is very slow due to the high degree both theoretically and in practice.  
However, the algorithm to find the decomposition of the persistence module is designed for modules over very general algebras.
Significant improvements should be possible that utilize the fact that commutative $G$-modules has additional structure.
The heuristics used in Section~\ref{sec:applications} used this additional structure, but we have been unable to prove that they will always be successful in correctly decomposing the persistence module.

Recent work has focused on developing parallel frameworks for persistence~\cite{DBLP:journals/corr/BauerKR13}.  
A second practical application of our formulation would be using the type of splitting shown is 
Section~\ref{subsec:ex} as a basis for a divide and conquer algorithm that allows parallel 
computation of standard persistence.  For example, in Figure~~\ref{fig:example}, we can calculate 
the persistence of the entire space by combining the computation for the subspaces $X,Y,$ and $Z$, which can be done in parallel.

There is the potential to improve our algorithm for all single-source single- sink subgraphs
using tools from recent work that computes zigzag homology in matrix multiply 
time~\cite{zigzag-mm}.  Consider a directed tree $T$ in the DAG for a single 
sink $s$; this tree is composed of paths which look like filtrations used in zigzag persistence, so we can 
use the recent zigzag algorithm for each path.  If this tree has common subpaths from previous 
calls to the zigzag algorithm, we could potentially speed up our algorithm by using this 
information.  On the other hand, if the tree has no such common subpaths, then 
intuitively we should be able to balance the fact that the disjoint 
paths sum to $n$ (so that we either have many very short paths or the tree 
is a single path).   Balancing this recursion based on the analysis in~\cite{zigzag-mm}, however,  
gives no better running time than the naive $O(n^2 M(l))$ algorithm we briefly outlined in 
Section~\ref{sec:single}, 
since the time to combine the information from a previous call with the new recursive call 
will take longer than the call itself.  A better algorithm that uses this information 
successfully is an interesting direction to consider.

Finaly, the alignment process proposed in Section~\ref{subsec:comparison} should be studied further.
It could prove a faster alternative to finding interleavings of filtrations.

\section*{Acknowledgments}

The authors would like to thank Afra Zomorodian for suggesting this problem during a visit, 
as well as for additional comments and suggestions along the way.  
We would also like to thank Greg Marks and Michael May for helpful conversations 
during the course of this work.
Finally, we would like to thank the anonymous reviewers for a number of suggestions for improvement of this paper.


\bibliography{dag-persistence}

\end{document}